\documentclass[12pt]{article}
\usepackage{lMac}
\usepackage{becMac}

\renewcommand{\sd}{\mathrm{d}}

\def\showintremarks{n}   
\def\showissues{n}       

\begin{document}
\title{Bloch Theory for Periodic Block Spin Transformations}

\author{Tadeusz Balaban}
\affil{\small Department of Mathematics \authorcr
       Rutgers, The State University of New Jersey \authorcr
       tbalaban@math.rutgers.edu\authorcr
       \  }

\author{Joel Feldman\thanks{Research supported in part by the Natural 
                Sciences and Engineering Research Council 
                of Canada and the Forschungsinstitut f\"ur 
                Mathematik, ETH Z\"urich.}}
\affil{Department of Mathematics \authorcr
       University of British Columbia \authorcr
       feldman@math.ubc.ca \authorcr
       http:/\hskip-3pt/www.math.ubc.ca/\squig feldman/\authorcr
       \  }

\author{Horst Kn\"orrer}
\author{Eugene Trubowitz}
\affil{Mathematik \authorcr
       ETH-Z\"urich \authorcr
       knoerrer@math.ethz.ch, trub@math.ethz.ch \authorcr
       http:/\hskip-3pt/www.math.ethz.ch/\squig knoerrer/}


\maketitle

\begin{abstract}
\noindent
Block spin renormalization group is the main tool  used in our  program 
to see symmetry breaking in a weakly interacting many Boson system 
on a three dimensional lattice at low temperature. 
It generates operators, like the fluctuation integral covariance, that act 
on some lattice but are translation invariant only with respect to a proper 
sublattice. This paper constructs a Bloch/Floquet framework that is 
appropriate for bounding such operators.

\end{abstract}

\newpage
\tableofcontents

\newpage
\subsection{Introduction}

One standard implementation of the renormalization group philosophy 
\cite{Wil} uses block spin transformations. See \cite{KAD,BalLausane,GK,BalPalaiseau,Dim1,BlockSpin}.
Concretely, suppose  we are to control a functional integral on a finite\footnote{Usually, 
the finite lattice is a ``volume cutoff'' infinite lattice and one wants 
to get bounds that are uniform in the size of the volume cutoff.} 
lattice $\cX_-$ of the form
\begin{equation}\label{introbstrbasiscfi}
\int \smprod_{x\in\cX_-} \sfrac{d\phi^*(x) d\phi(x)}{2\pi i}\,
 e^{A(\al_1,\cdots,\al_s;\phi^*,\phi)}
 \end{equation}
with an action $A(\al_1,\cdots,\al_s;\phi_*,\phi)$ that is a function 
of external complex valued fields $\al_1$, $\cdots$, $\al_s$,
and the two\footnote{ In the actions, we  treat $\phi$ and its complex conjugate $\phi^*$ as independent variables.} complex fields $\phi_*,\phi$ on $\cX_-$. This scenario occurs in \cite{PAR1,PAR2}, where  we use  block spin renormalization group
maps to exhibit the formation of a potential well, signalling the onset of symmetry breaking in
a many particle system of weakly interacting Bosons in three space dimensions. 
(For an overview, see \cite{ParOv}.)

Under the  renormalization group approach to controlling integrals like
\eqref{introbstrbasiscfi} one successively ``integrates out'' 
lower and lower energy degrees of freedom. In the block spin formalism this 
is implemented by considering a  decreasing sequence of sublattices of $\cX_-$.  
The formalism produces, for each such sublattice, a representation of the
integral \eqref{introbstrbasiscfi} that is a functional integral whose
integration variables are indexed by that sublattice. To pass from the representation
associated with one sublattice $\cX\subset\cX_-$, with integration variables $\psi(x)$, $x\in\cX$, to the 
representation associated to the next coarser sublattice $\cX_+\subset\cX$,
with integration variables $\th(y)$, $y\in\cX_+$, one 
\begin{itemize}[leftmargin=*, topsep=2pt, itemsep=0pt, parsep=0pt]
\item 
paves $\cX$ by rectangles centered at the points of $\cX_+$ and then, 
\item 
for each $y\in\cX_+$ integrates out all values of $\psi$ whose ``average 
value'' over the rectangle centered at $y$ is equal to $\th(y)$.
The precise ``average value'' used is determined by an averaging profile
$q$. As in \eqref{eqnBOaveop}, one uses this profile to define an averaging operator 
$Q$ from  the space $\cH$ of fields on 
 $\cX$ to the space $\cH_+$ of fields on  $\cX_+$. One then implements the
``integrating out'' by first, inserting, into the integrand, $1$ expressed as a 
constant times the Gaussian integral
\begin{align*}
&\int \smprod_{y\in\cX_+} \sfrac{d\th^*(y) d\th(y)}{2\pi i}
                                        e^{-b\< \th^*-Q\,\psi_*\,,\, \th-Q\,\psi) \>} 
\end{align*}
with some constant $b>0$, and then interchanging the order of the $\th$ and
$\psi$ integrals.
\end{itemize}
For example, in \cite{ParOv,PAR1,PAR2} the model is initially formulated as a functional
integral with integration variables indexed by a lattice\footnote{The volume cutoff 
is determined by $L_\tp$ and $L_\sp$.}
$\big(\bbbz/L_\tp\bbbz\big)\times\big(\bbbz^3/L_\sp\bbbz^3\big)$.
After $n$ renormalization group steps this lattice is scaled down
to $\cX_n= \big(\sfrac{1}{L^{2n}}\bbbz \big/ \sfrac{L_\tp}{L^{2n}}\bbbz\big)
   \times \big(\sfrac{1}{L^n}\bbbz^3 \big/\sfrac{L_\sp}{L^n}\bbbz^3 \big)$. 
The decreasing family of sublattices is
$\cX_{j}^{(n-j)} 
= \big(\sfrac{1}{L^{2j}}\bbbz \big/ \sfrac{L_\tp}{L^{2n}}\bbbz\big)
   \times \big(\sfrac{1}{L^{j}}\bbbz^3 \big/
                      \sfrac{L_\sp}{L^n}\bbbz^3 \big)$, 
$j=n$, $n-1$, $\cdots$.  The abstract lattices
$\cX_-$, $\cX$, $\cX_+$ in the above framework correspond to
$\cX_n$, $\cX_0^{(n)}$ and $\cX_{-1}^{(n+1)}$, respectively.

In this framework there are a good number of
linear operators that act on functions defined on a finite lattice
and that are translation invariant with respect to a sub--lattice.
For example the block spin averaging operator 
$Q$ above (which is an abstraction of the operator $Q$ of 
\cite[Definition \defHTblockspintr.a]{PAR1} and \cite[(\eqnPBSaveop)]{POA})
acts on functions defined on the lattice $\cX$, but is translation invariant
only with respect to the sublattice $\cX_+$.
Similarly the operator $Q_-$ of  \cite[(\introbstrstepnfi)]{BlockSpin}
(which is an abstraction of the operator $Q_n$ of 
\cite[Definition \defHTbackgrounddomaction.a]{PAR1} and \cite[(\eqnPBSqn)]{POA})
acts on functions defined on the lattice $\cX_-$, but is translation invariant
only with respect to the sublattice $\cX$.
As another example, the fluctuation integral covariance
$C$ of \cite[(\eqnBSdefCascovariance)]{BlockSpin}
(which is an abstraction of the operator $C^{(n)}$ of 
\cite[(\eqnHTcn)]{PAR1} and \cite[\S\secPOcovariance]{POA})
acts on functions defined on the lattice $\cX$, but is translation invariant
only with respect to the sublattice $\cX_+$.
In this paper, we use the Bloch/Floquet theory (see, for example, \cite{Kuchment}) approach to 
develop some general 
machinery for bounding such linear operators. 
In \cite{PAR1,PAR2} the operators of interest tend to be periodizations of operators
acting on $L^2$ of an infinite lattice. An important example is the 
``differential'' operator $D_n$. See \cite[Remark \remPDOftDn.a]{POA}. 
We also develop general machinery for bounding such periodizations.  
In \cite{POA} we use the results of this paper to bound many of the operators appearing in
\cite{PAR1,PAR2}.

\subsection{Periodic Operators in ``Position Space'' and ``Momentum Space'' Environments}
We start by setting up a general environment consisting of a ``fine''
lattice and a ``coarse'' sub--lattice. We shall consider operators that 
act on functions defined on the former and that are translation invariant with 
respect to the latter.   Let $\veps_T,\veps_X>0$, $L_T,L_X\in\bbbn$ 
and $\cL_T\in L_T\bbbn$, $\cL_X\in L_X\bbbn$ and define the (finite) 
$(\sd+1)$--dimensional lattices
\begin{align*}
\cX_\fin&= \big(\veps_T\bbbz/\veps_T\cL_T\bbbz\big)\times
            \big(\veps_X\bbbz^\sd/\veps_X\cL_X\bbbz^\sd)
\\
\cX_\crs&=\big(L_T\veps_T\bbbz/\veps_T\cL_T\bbbz\big)\times
                  \big(L_X\veps_X\bbbz^\sd/\veps_X\cL_X\bbbz^\sd)
\end{align*}
and the corresponding Hilbert spaces
\begin{align*}
\cH_f&=L^2(\cX_\fin) \qquad
     & \<\phi^*_1,\phi_2\>_f
        &=\vol_f\sum_{u\in\cX_\fin}\phi_1(u)^*\phi_2(u)\\
\cH_c&=L^2(\cX_\crs)& 
\<\psi_1^*,\psi_2\>_c
&=\vol_c\sum_{x\in\cX_\crs}\psi_1(x)^*\psi_2(x)
\end{align*}
where we use
\begin{equation*}
\vol_f = \veps_T\veps_X^\sd\qquad
\vol_c = (\veps_TL_T)(\veps_XL_X)^\sd\qquad
\end{equation*}
to denote the volume of a single cell in $\cX_\fin$, 
and $\cX_\crs$, respectively. For the Bloch construction,
it will also be useful to define the ``single period'' lattice
\begin{equation*}
\cB=\big(\veps_T\bbbz/L_T\veps_T\bbbz\big)\times
                  \big(\veps_X\bbbz^\sd/L_X\veps_X\bbbz^\sd)
\cong \cX_\fin/\cX_\crs
\end{equation*}

The lattices dual to $\cX_\fin$, $\cX_\crs$ and 
$\cB$ are
\begin{align*}
\hat\cX_\fin
&= \big(\sfrac{2\pi}{\veps_T\cL_T}\bbbz/\sfrac{2\pi}{\veps_T}\bbbz\big)
 \times
  \big(\sfrac{2\pi}{\veps_X\cL_X}\bbbz^\sd/\sfrac{2\pi}{\veps_X}\bbbz^\sd\big)
\\
\hat\cX_\crs&= \big(\sfrac{2\pi}{\veps_T\cL_T}\bbbz/\sfrac{2\pi}{L_T\veps_T}\bbbz\big)
       \times
\big(\sfrac{2\pi}{\veps_X\cL_X}\bbbz^\sd/\sfrac{2\pi}{L_X\veps_X}\bbbz^\sd\big)
\\
\hat\cB&= 
\big(\sfrac{2\pi}{L_T\veps_T}\bbbz/\sfrac{2\pi}{\veps_T}\bbbz\big)
   \times
\big(\sfrac{2\pi}{L_X\veps_X}\bbbz^\sd/\sfrac{2\pi}{\veps_X}\bbbz^\sd\big)
\cong \hat\cX_\crs/\hat\cX_\fin
\end{align*}
We denote by
\begin{equation*}
\hat\pi:\hat\cX_\fin\rightarrow \hat\cX_\crs
\end{equation*}
the canonical projection from $\hat\cX_\fin$ to 
$\hat\cX_\crs$. It has kernel $\hat\cB$.
Observe that
\begin{equation*}
p\cdot x=\hat\pi(p)\cdot x\ {\rm mod}\, 2\pi\qquad
\text{for all }x\in \cX_\crs,\ p\in \hat\cX_\fin
\end{equation*}
The Fourier and inverse Fourier transforms are, for
$\phi\in\cH_f$, $\psi\in\cH_c$, $\ze\in L^2(\cB)$,
$p\in\hat\cX_\fin$, $k\in \hat\cX_\crs$, 
$\ell\in\hat\cB$, $u\in\cX_\fin$,
$x\in\cX_\crs$ and $w\in\cB$,
\begin{align*}
\hat \phi(p)
 &=\vol_f\sum\limits_{u\in\cX_\fin} \phi(u) e^{-i p\cdot u }
\quad &
\phi(u)&=\sfrac{\hvol_f}{(2\pi)^{1+\sd}}
    \sum_{p\in\hat\cX_\fin}  \hat \phi(p) e^{i u\cdot p }
\\
\hat \psi(k)&=\vol_c\sum_{x\in\cX_\crs} 
                      \psi(x) e^{-i k\cdot x }
\quad &
\psi(x)
&=\sfrac{\hvol_c}{(2\pi)^{1+\sd}}\sum_{k\in\hat\cX_\crs} 
                            \hat \psi(k) e^{i k\cdot x }
\\
\hat \ze(\ell)
&=\vol_f\sum\limits_{w\in\cB} \ze(w) e^{-i \ell\cdot w }
\quad &
\ze(w)&=\sfrac{\hvol_b}{(2\pi)^{1+\sd}}\sum_{\ell\in\hat\cB} 
                            \hat \ze(\ell) e^{i w\cdot \ell }
\end{align*}
where
\begin{equation*}
\hvol_f = \sfrac{(2\pi)^{1+\sd}}{(\veps_T\cL_T)(\veps_X\cL_X)^\sd}\qquad
\hvol_c = \sfrac{(2\pi)^{1+\sd}}{(\veps_T\cL_T)(\veps_X\cL_X)^\sd}\qquad
\hvol_b = \sfrac{(2\pi)^{1+\sd}}{(\veps_TL_T)(\veps_XL_X)^\sd}
\end{equation*}
denote the volume of a single cell in $\hat\cX_\fin$, 
$\hat\cX_\crs$ and $\hat\cB$, respectively. Observe that
\begin{equation}\label{eqnBOvolhvol}
\begin{alignedat}{3}
\sfrac{\vol_f\hvol_f}{(2\pi)^{1+\sd}}
&=\sfrac{1}{\cL_T\cL_X^\sd}=\sfrac{1}{|\cX_\fin|}
&=\sfrac{1}{|\hat\cX_\fin|}\\
\sfrac{\vol_c\hvol_f}{(2\pi)^{1+\sd}}=\sfrac{\vol_c\hvol_c}{(2\pi)^{1+\sd}}
&=\sfrac{L_TL_X^\sd}{\cL_T\cL_X^\sd}=\sfrac{1}{|\cX_\crs|}
&=\sfrac{1}{|\hat\cX_\crs|}
\end{alignedat}
\end{equation}
where $|\cX_\crs|$ denotes the number of points in 
$\cX_\crs$. By \eqref{eqnBOvolhvol} and the fact that
$\de_{u,u'}=\sfrac{1}{|\hat\cX_\fin|}
        \sum_{p\in\hat\cX_\fin} e^{ip\cdot u}e^{-ip\cdot u'}$, 
\begin{alignat*}{3}
\<\phi_1,\phi_2\>_f
        &=\vol_f\sum_{u\in\cX_\fin}\phi_1(u)\phi_2(u)
        &=\sfrac{\hvol_f}{(2\pi)^{1+\sd}}
              \sum_{p\in\hat\cX_\fin}
                    \hat\phi_1(-p)\hat\phi_2(p)
\\
\<\psi_1,\psi_2\>_c
&=\vol_c\sum_{x\in\cX_\crs}\psi_1(x)\psi_2(x)
        &=\sfrac{\hvol_c}{(2\pi)^{1+\sd}}
              \sum_{k\in\hat\cX_\crs}
                    \hat\psi_1(-k)\hat\phi_2(k)
\end{alignat*}


Let $A$ be any operator on $\cH_f$ that is translation invariant with
respect to $\cX_\crs$.  We call such an operator a ``periodic
operator''. Denote by $A(u,u')$ its kernel, defined so that
\begin{equation*}
(A\phi)(u)=\vol_f\sum_{u'\in\cX_\fin}A(u,u')\phi(u')
\end{equation*}
By ``translation invariant with respect to $\cX_\crs$'', we mean
that $A(u+x,u'+x)=A(u,u')$ for all $u,u'\in\cX_\fin$ and 
$x\in\cX_\crs$.
Set\footnote{The ``normal prefactor'' for $\hat A$ would be 
$\vol_f^2$. We have chosen $\sfrac{\vol_f}{|\cX_\fin|}
=\sfrac{\hvol_f}{(2\pi)^{1+\sd}}\vol_f^2$ so as to replace approximate 
Dirac $(2\pi)^{1+\sd}\de(p-p')$'s with simple Kronecker $\de_{p,p'}$'s
in the translation invariant case.}, 
for $p,p'\in\hat\cX_\fin$,
\begin{equation}\label{eqnBOopft}
\hat A(p,p')=\sfrac{\vol_f}{|\cX_\fin|}\sum_{u,u'\in\cX_\fin}
                 e^{-ip\cdot u}A(u,u')e^{ip'\cdot u'}
\end{equation}
and, for $u,u'\in\cX_\fin$ and $\rk\in
\sfrac{2\pi}{\veps_T\cL_T}\bbbz\times\sfrac{2\pi}{\veps_X\cL_X}\bbbz^\sd$,
the ``universal cover'' of $\hat\cX_\crs$,
\begin{equation}\label{eqnBOkerk}
A_\rk(u,u')=\vol_c\hskip-10pt
      \sum_{\atop{u''\in\cX_\fin}{u''-u'\in\cX_\crs}}
      \hskip-10pt
      e^{-i\rk\cdot u}A(u,u'')e^{i\rk\cdot u''}
\end{equation}
For each fixed $u,u'\in\cX_\fin$, $\rk\mapsto A_\rk(u,u')$
is not a function on the torus $\hat\cX_\crs$ since,
for $\rp\in\sfrac{2\pi}{L_T\veps_T}\bbbz
           \times \sfrac{2\pi}{L_X\veps_X}\bbbz^\sd$,
\begin{equation*}
A_{\rk+\rp}(u,u')=e^{-i\rp(u-u')}A_\rk(u,u')
\end{equation*}
This is why we defined $A_\rk$ for $\rk\in\sfrac{2\pi}{\veps_T\cL_T}\bbbz
\times\sfrac{2\pi}{\veps_X\cL_X}\bbbz^\sd$, rather than
$\rk\in\hat\cX_\crs$.
On the other hand, $k\mapsto e^{ik(u-u')} A_k(u,u')$ is a well--defined 
function on $\hat\cX_\crs$.

The following lemma is standard.
\begin{lemma}\label{lemBOkervar}
Let $A$ be an operator on $\cH_f$ that is translation invariant with
respect to $\cX_\crs$. 
\begin{enumerate}[label=(\alph*), leftmargin=*]
\item $A(u,u')=\sfrac{\hvol_f}{(2\pi)^{1+\sd}}
           {\displaystyle\sum_{p,p'\in\hat\cX_\fin}}
                 e^{ip\cdot u}\hat A(p,p')e^{-ip'\cdot u'}$

\item 
$A(u,u')=\sfrac{\hvol_c}{(2\pi)^{1+\sd}}
       {\displaystyle\sum_{\atop{[\rk]\in\hat\cX_\crs}
                                {\ell,\ell'\in\hat\cB}}}
                 e^{i\ell\cdot u}\hat A(\rk+\ell,\rk+\ell')e^{-i\ell'\cdot u'}
                 e^{i\rk\cdot (u-u')}
$

\noindent
Here $\sum\limits_{[\rk]\in\hat\cX_\crs}f(\rk)$ means that one
sums $\rk$ over a subset of $\sfrac{2\pi}{\veps_T\cL_T}\bbbz\times\sfrac{2\pi}{\veps_X\cL_X}\bbbz^\sd$ that contains exactly one
(arbitrary) representative from each equivalence class of $\hat\cX_\crs$.
Note that if $\rk\in\sfrac{2\pi}{\veps_T\cL_T}\bbbz\times\sfrac{2\pi}{\veps_X\cL_X}\bbbz^\sd$ and $\ell\in\hat\cB$,
then $\rk+\ell\in\hat\cX_\fin$.

\item  
$\widehat{(A\phi)}(p)
={\displaystyle\sum_{p'\in\hat\cX_\fin}}\hat A(p,p')\hat\phi(p')$
for all $\phi\in \cH_f$.

\item 
For each $\rk\in\sfrac{2\pi}{\veps_T\cL_T}\bbbz
\times\sfrac{2\pi}{\veps_X\cL_X}\bbbz^\sd$, $A_\rk(u,u')$ is 
periodic with respect to $\cX_\crs$ 
in both $u$ and $u'$ and
\begin{equation*}
A(u,u')=\sfrac{\hvol_c}{(2\pi)^{1+\sd}}
            {\displaystyle\sum_{k\in\hat\cX_\crs}}
                 e^{ik\cdot u}A_k(u,u')e^{-ik\cdot u'}
\end{equation*}

\item
$
A_\rk(u,u')={\displaystyle
           \sum_{\ell,\ell'\in\hat\cB}}
                 e^{i\ell\cdot u}\hat A(\rk+\ell,\rk+\ell')e^{-i\ell'\cdot u'}
$

\item 
Define the transpose of $A$ by $A^\star(u,u')=A(u',u)$. Then 
\begin{equation*}
A^\star(u,u')=\sfrac{\hvol_c}{(2\pi)^{1+\sd}}
       {\displaystyle\sum_{\atop{[\rk]\in\hat\cX_\crs}
                                {\ell,\ell'\in\hat\cB}}}
                 e^{i\ell\cdot u}\hat A(-\rk-\ell',-\rk-\ell)e^{-i\ell'\cdot u'}
                 e^{i\rk\cdot (u-u')}
\end{equation*}
\end{enumerate}
\end{lemma}

\subsection{Periodized Operators}\label{secBOperiodization}
Define the (infinite) lattices
\begin{align*}
\cZ_\fin= \veps_T\bbbz\times\veps_X\bbbz^\sd
\qquad
\cZ_\crs=L_T\veps_T\bbbz\times L_X\veps_X\bbbz^\sd
\cr
\end{align*}

\begin{definition}[Periodization]\label{defBOperiodization}
Suppose that $a(\ru,\ru')$ is a function on 
$\cZ_\fin\times\cZ_\fin$ that
\begin{itemize}[leftmargin=*, topsep=2pt, itemsep=0pt, parsep=0pt]
\item 
is translation invariant with respect to 
$\cZ_\crs$ in the sense that $a(\ru+\rx,\ru'+\rx)=a(\ru,\ru')$ 
for all $\rx\in\cZ_\crs$ and $\ru,\ru'\in \cZ_\fin$ and
\item 
has finite $L^1$--$L^\infty$ norm (i.e. 
$\sup\limits_{\ru\in\cZ_\fin}\ \sum\limits_{\ru'\in\cZ_\fin}
          \big|a(\ru,\ru')\big|$ and 
$\sup\limits_{\ru'\in\cZ_\fin}\ \sum\limits_{\ru\in\cZ_\fin}
          \big|a(\ru,\ru')\big|$ are both finite)
\end{itemize}
and that the operator $A$ (on $\cH_f$) acts by
\begin{equation}\label{eqnBOrestricttotorus}
(A\phi)([\ru])=\vol_f\sum_{\ru'\in \cZ_\fin} a(\ru,\ru')\phi([\ru'])
\end{equation}
Here, for each $\ru\in\cZ_\fin$, the notation $[\ru]$
means the equivalence class in $\cX_\fin$ that contains $\ru$. Then 
we say that $A$ is the periodization of $a$. It is ``$a$ with periodic
boundary conditions on a box of size $\veps_T\cL_T\times\overbrace{\veps_X\cL_X
\times\cdots\times\veps_X\cL_X}^{\sd{\rm\ factors}}\ $''.
\end{definition}

\pagebreak[2]
\begin{remark}\label{remBOperiodization}
\ 
\begin{enumerate}[label=(\alph*), leftmargin=*]
\item 
The right hand side of \eqref{eqnBOrestricttotorus} is independent 
of the representative $\ru$ chosen from $[\ru]$ (by translation invariance 
with respect to  
$\veps_T\cL_T\bbbz\times\veps_X\cL_X\bbbz^\sd\subset\cZ_\crs$).

\item The kernel of $A$ is given by 
\begin{equation*}
A\big([\ru],[\ru']\big)=\sum_{\atop{\ru''\in\cZ_\fin}{[\ru'']=[\ru']}}
                a(\ru,\ru'')
=\sum_{\rz\in\veps_T\cL_T\bbbz\times\veps_X\cL_X\bbbz^\sd}\hskip-10pt
                a(\ru,\ru'+\rz)
\end{equation*}
The sum converges because $a$ has finite $L^1$--$L^\infty$ norm.
This is the motivation for the name the ``periodization of $a$''.

\item 
If $A$ is the periodization of $a$ and $B$ is the periodization
of $b$, then $C=AB$ is the periodization of 
\begin{equation*}
c(\ru,\ru')=\vol_f\sum_{\ru''\in \cZ_\fin} a(\ru,\ru'')b(\ru'',\ru')
\end{equation*}

\end{enumerate}
\end{remark}

\noindent
Let $\hat\cZ_\crs=\big(\bbbr/\sfrac{2\pi}{\veps_TL_T}\bbbz\big)
\times \big(\bbbr^\sd/\sfrac{2\pi}{\veps_XL_X}\bbbz^\sd\big)$ 
be the dual space of $\cZ_\crs$. Its universal cover is $\bbbr\times\bbbr^\sd$.
For each $\rk\in\bbbr\times\bbbr^\sd$, set, for $\ru,\ru'\in\cZ_\fin$,
\begin{equation}\label{eqnBOifkerk}
a_\rk(\ru,\ru')=\vol_c\hskip-10pt
      \sum_{\atop{\ru''\in\cZ_\fin}{\ru''-\ru'\in\cZ_\crs}}
      \hskip-10pt
      e^{-i\rk\cdot\ru}a(\ru,\ru'')e^{i\rk\cdot \ru''}
\end{equation}
and, for $\ell,\ell'\in\hat\cB$,
\begin{equation}\label{eqnBOifkerkell}
\begin{split}
\hat a_\rk(\ell,\ell')
&=\sfrac{1}{|\cB|^2}
    \sum_{[\ru],[\ru']\in\cB}  \hskip-5pt
      e^{-i\ell\cdot\ru } a_\rk(\ru,\ru')e^{i\ell'\cdot\ru' } 
     \\
&=\sfrac{\vol_f}{|\cB|}
    \sum_{\atop{[\ru]\in\cB}{\ru'\in\cZ_\fin}}  \hskip-5pt
      e^{-i\ell\cdot\ru } a(\ru,\ru')e^{i\ell'\cdot\ru' }
      e^{-i\rk\cdot(\ru- \ru')} 
\end{split}
\end{equation}
(Recall that $\sfrac{1}{|\cB|^2}
=\sfrac{\vol_f}{\vol_c|\cB|}$. We shall show in Lemma
\ref{lemBOifkervar}.a, below, that $a_\rk(\ru,\ru')$ is periodic with respect 
to $\cZ_\crs$ in both $\ru$ and $\ru'$.) By the $L^1$--$L^\infty$ 
hypothesis and the Lebesgue dominated convergence theorem, both 
$a_\rk(\ru,\ru')$ and $a_\rk(\ell,\ell')$ are continuous in $\rk$.

\begin{remark}\label{remBOatwisted}
As was the case for $A_k(u,u')$, for each fized $\ru,\ru'\in\cZ_\fin$, 
the map $\rk\mapsto a_\rk(\ru,\ru')$ is not a function on the torus
$\hat\cZ_\crs$ since, for $\rp\in \sfrac{2\pi}{\veps_TL_T}\bbbz
\times \sfrac{2\pi}{\veps_XL_X}\bbbz^\sd$,
$$
a_{\rk+\rp}(\ru,\ru')=e^{-i\rp\cdot(\ru-\ru')}a_\rk(\ru,\ru')
$$
However
$$
\rk\in\hat\cZ_\crs\mapsto
e^{i\rk\cdot(\ru-\ru')}a_\rk(\ru,\ru')
=\vol_c\hskip-5pt
      \sum_{\rx\in\cZ_\crs}
      \hskip-5pt
      a(\ru,\ru'+\rx)e^{i\rk\cdot\rx}
$$
is a legitimate function on the torus $\hat\cZ_\crs$ and is in fact 
the Fourier transform of the function
$$
\rx\in\cZ_\crs\mapsto a(\ru,\ru'+\rx)
$$
Correspondingly, for $\rp\in \sfrac{2\pi}{\veps_TL_T}\bbbz
\times \sfrac{2\pi}{\veps_XL_X}\bbbz^\sd$ and 
$\ell,\ell'\in\hat\cB$
$$
\hat a_{\rk+\rp}(\ell,\ell')=\hat a_{\rk}(\ell+\rp,\ell'+\rp)
$$
\end{remark}

The following two lemmas are again standard.

\begin{lemma}\label{lemBOifkervar}
Let $a(\ru,\ru'):\cZ_\fin\times\cZ_\fin\rightarrow\bbbc$ 
obey the conditions of Definition \ref{defBOperiodization}, 
and, in particular, be translation invariant with respect to $\cZ_\crs$.
\begin{enumerate}[label=(\alph*), leftmargin=*]
\item
For each $\rk\in\bbbr\times\bbbr^\sd$, $a_\rk(\ru,\ru')$ is 
periodic with respect to $\cZ_\crs$ 
in both $\ru$ and $\ru'$ and
\begin{align*}
a(\ru,\ru')&=\int_{\hat\cZ_\crs}\hskip-10pt
                a_\rk(\ru,\ru') e^{i\rk\cdot (\ru-\ru') }
                \sfrac{d^{1+\sd}\rk}{(2\pi)^{1+\sd}} \\
&=\sum_{\ell,\ell'\in\hat\cB}
\int_{\hat\cZ_\crs}\hskip-10pt
                e^{i\ell\cdot \ru }\hat a_\rk(\ell,\ell')e^{-i\ell'\cdot \ru' }
                e^{i\rk\cdot (\ru-\ru') }  
                \sfrac{d^{1+\sd}\rk}{(2\pi)^{1+\sd}}
\end{align*}

\item 
If, in addition, $a(\ru,\ru')=\al(\ru-\ru')$ is translation invariant
with respect to $\cZ_\fin$, then
\begin{equation*}
\hat a_\rk(\ell,\ell')
=\de_{\ell',\ell}\ \hat\al(\rk+\ell)
\end{equation*}
where $\hat\al(\rp)
    =\vol_f\hskip-5pt\sum\limits_{\ru\in\cZ_\fin}  \hskip-5pt
      \al(\ru)e^{-i\rp\cdot\ru}$.

\item 
Let $A$ be the periodization of $a$. Then
\begin{align*}
A_\rk([\ru],[\ru'])&=a_\rk(\ru,\ru')
\\
\hat A(\rk+\ell,\rk+\ell')
   &=\hat a_\rk(\ell,\ell')
\end{align*}
for all $\rk\in\sfrac{2\pi}{\veps_T\cL_T}\bbbz
\times\sfrac{2\pi}{\veps_X\cL_X}\bbbz^\sd$,
$[\ru],[\ru']\in\cX_\fin$ and $\ell,\ell'\in\hat\cB$.
\end{enumerate}
\end{lemma}

\begin{lemma}\label{lemBOperiodalg}
\ 
\begin{enumerate}[label=(\alph*), leftmargin=*]
\item 
If $a(\ru,\ru')=\sfrac{1}{\vol_f}\de_{\ru,\ru'}$ is the kernel
of the identity operator,
then
$
\hat a_\rk(\ell,\ell') =\de_{\ell,\ell'}
$.
\item 
Let $a(\ru,\ru'),b(\ru,\ru')
:\cZ_\fin\times\cZ_\fin\rightarrow\bbbc$ 
both obey the conditions of Definition \ref{defBOperiodization}, and set
\begin{equation*}
c(\ru,\ru')=\vol_f\sum_{\ru''\in\cZ_\fin}a(\ru,\ru'')b(\ru'',\ru')
\end{equation*}
Then, for all $\rk\in\sfrac{2\pi}{\veps_T\cL_T}\bbbz
  \times\sfrac{2\pi}{\veps_X\cL_X}\bbbz^\sd$ and 
$\ell,\ell'\in\hat\cB$,
\begin{equation*}
c_\rk(\ell,\ell')
=\sum_{\ell''\in\hat\cB}a_\rk(\ell,\ell'')b_\rk(\ell'',\ell')
\end{equation*}
\end{enumerate}
\end{lemma}

We now generalize the above discussion to include periodized operators 
from $L^2(\cX_\crs)$ to $L^2(\cX_\fin)$ and vice versa.
If $b(\ru,\rx)$ and $c(\rx,\ru)$ are translation invariant with respect to
$\cZ_\crs$ (with $\rx$ running over $\cZ_\crs$
and with $\ru$ running over $\cZ_\fin$ as usual)
and have finite $L^1$--$L^\infty$ norm, we define,
for $\rk\in\bbbr\times\bbbr^\sd$ and $\ell,\ell'\in\hat\cB$,
\begin{equation}\label{eqnBOasymmetric}
\begin{split}
\hat b_\rk(\ell)= \vol_f
    \sum_{[\ru]\in\cB\atop\rx\in\cZ_\crs}  \hskip-5pt
      e^{-i(\rk+\ell)\cdot\ru } b(\ru,\rx)
      e^{i\rk\cdot\rx} 
\quad
\hat c_\rk(\ell')= \vol_f
    \sum_{[\ru]\in\cB\atop\rx\in\cZ_\crs}  \hskip-5pt
      e^{-i\rk\cdot\rx} c(\rx,\ru)e^{i(\rk+\ell')\cdot\ru } 
\end{split}
\end{equation}   
For $\rp\in \sfrac{2\pi}{\veps_TL_T}\bbbz
\times \sfrac{2\pi}{\veps_XL_X}\bbbz^\sd$ and 
$\ell,\ell'\in\hat\cB$
\begin{equation*}
\hat b_{\rk+\rp}(\ell)=\hat b_{\rk}(\ell+\rp)\qquad
\hat c_{\rk+\rp}(\ell')=\hat c_{\rk}(\ell'+\rp)
\end{equation*}
The inverse transforms are
\begin{align*}
b(\ru,\rx)&=\sum_{\ell\in\hat\cB}
\int_{\hat\cZ_\crs}\hskip-10pt
                e^{i\ell\cdot \ru }\hat b_\rk(\ell)
                e^{i\rk\cdot (\ru-\rx) }  
                \sfrac{d^{1+\sd}\rk}{(2\pi)^{1+\sd}} \\
c(\rx,\ru)&=\sum_{\ell'\in\hat\cB}
\int_{\hat\cZ_\crs}\hskip-10pt
                \hat a_\rk(\ell,\ell')e^{-i\ell'\cdot \ru }
                e^{i\rk\cdot (\rx-\ru) }  
                \sfrac{d^{1+\sd}\rk}{(2\pi)^{1+\sd}} 
\end{align*}
For $\psi\in L^2(\cX_\crs)$ and $\phi\in L^2(\cX_\fin)$
\begin{equation}\label{eqnPOftaction}
\begin{split}
\widehat{b\psi}(k+\ell)&= \hat b_k(\ell)\hat\psi(k)\\
\widehat{c\phi}(k)&= \sum_{\ell'\in\hat\cB}\hat c_k(\ell')\hat\phi(k+\ell')
\end{split}
\end{equation}
If $b^\star(\rx,\ru)=b(\ru,\rx)$ and $c^\star(\ru,\rx)=c(\rx,\ru)$ are the transposes
of $b$ and $c$, respectively, then
\begin{equation}\label{eqnPOtranspose}
\hat b^\star_\rk(\ell')=\hat b_{-\rk}(-\ell')
\qquad
\hat c^\star_\rk(\ell)=\hat c_{-\rk}(-\ell)
\end{equation}

\subsection{Averaging Operators}

In this subsection, we analyze ``averaging operators'' as examples
of periodic operators. 
Fix a function $q:\cX_\fin\rightarrow\bbbr$ and define the ``averaging
operator'' $Q:\cH_f\rightarrow\cH_c$ by
\begin{equation}\label{eqnBOaveop}
(Q\phi)(x)=\vol_f\sum_{u\in\cX_\fin}q(x-u)\phi(u)
\end{equation}
\begin{lemma}\label{lemBOQ}
\ 
\begin{enumerate}[label=(\alph*), leftmargin=*]
\item
The adjoint $Q^\star$ is given by
\begin{equation*}
\big(Q^\star\psi)(u)=
      \vol_c\sum_{x\in\cX_\crs} \psi(x)q(x-u)
\end{equation*}
\item
The composite operators $QQ^\star$ and $Q^\star Q$ are given by
\begin{align*}
(QQ^\star\psi)(x)&=\vol_f\vol_c
                     \sum_{u\in\cX_\fin\atop
                           x'\in\cX_\crs} q(x-u)q(x'-u)\psi(x')
\\
(Q^\star Q\phi)(u)&=\vol_f\vol_c
                     \sum_{u'\in\cX_\fin\atop
                           x\in\cX_\crs} q(x-u)q(x-u')\phi(u')
\end{align*} 
\end{enumerate}
\end{lemma}

\begin{proof} trivial. 
\end{proof}

\begin{example}\label{exBOnaive}
Assume that $L_T$  and $L_X$ are odd and choose $q$ to be 
$\sfrac{1}{\vol_c}$ times the characteristic
function of the rectangle $\veps_T\big[-\sfrac{L_T-1}{2}, \sfrac{L_T-1}{2} \big]
\times \veps_X\big[-\sfrac{L_X-1}{2}, \sfrac{L_X-1}{2} \big]^\sd$ in 
$\cX_\fin$. Observe that the number of points in this rectangle 
is exactly $L_TL_X^\sd$. For $x\in\cX_\crs$, denote by $\sq_x$ 
the rectangle $x+\veps_T\big[-\sfrac{L_T-1}{2}, \sfrac{L_T-1}{2} \big]
\times \veps_X\big[-\sfrac{L_X-1}{2}, \sfrac{L_X-1}{2} \big]^\sd$ in 
$\cX_\fin$. Also, for $u\in\cX_\fin$, let $\xi(u)$ be the point
of $\cX_\crs$ closest to $u$. Then
\begin{equation*}
(Q\phi)(x)=\sfrac{1}{L_TL_X^\sd}\sum_{u\in\sq_x}\phi(u)\qquad
\big(Q^\star\psi)(u)=\psi\big(\xi(u)\big)
\end{equation*}
The composite operators are
\begin{align*}
(QQ^\star\psi)(x)&=\sfrac{1}{L_TL_X^\sd}\sum_{u\in\sq_x}(Q^\star\psi)(u)
=\sfrac{1}{L_TL_X^\sd}\sum_{u\in\sq_x}\psi(x)
=\psi(x)
\\
(Q^\star Q\phi)(u)&=(Q\phi)\big(\xi(u)\big)
=\sfrac{1}{L_TL_X^\sd}\sum_{u'\in\sq_{\xi(u)}}\phi(u')
\end{align*} 
\end{example}

\begin{lemma}\label{lemBOfourier}
Let $Q:\cH_f\rightarrow\cH_c$ be the averaging operator of \eqref{eqnBOaveop},
but with $q:\cZ_\fin\rightarrow\bbbr$ and $q(\ru)$ vanishing unless
$|\ru_0|<\half\veps_T\cL_T$ and $|\ru_\nu|<\half\veps_X\cL_X$ for $\nu=1,2,3$.
\begin{enumerate}[label=(\alph*), leftmargin=*]
\item
For all $\phi\in\cH_f$ and $\psi\in\cH_c$,
\begin{align*}
\widehat{(Q\phi)}(k)
  &=\sum_{ \atop{p\in\hat\cX_\fin}{\hat\pi(p)=k}}
                       \hat q(p)\hat\phi(p)
&
\widehat{(Q^\star\psi)}(p)&=\overline{\hat q(p)}\,\hat\psi\big(\hat\pi(p)\big)
\\
\widehat{(QQ^\star\psi)}(k)
&=\Big(\smsum_{\atop{p\in\hat\cX_\fin}{\hat\pi(p)=k} }
                       \big|\hat q(p)\big|^2\Big)\hat\psi(k) &
\widehat{(Q^\star Q\phi)}(p)
&=\overline{\hat q(p)}\!\!\sum_{ \atop{p'\in\hat\cX_\fin}{\hat\pi(p')=\hat\pi(p)}}\!\!
                        \hat q(p')\,\hat\phi(p')
\end{align*}

\item
For $A=Q^\star Q$, 
\begin{align*}
\hat a_\rk(\ell,\ell')
&=\overline{q(\rk+\ell)}\,q(\rk+\ell')
\end{align*}
\end{enumerate}
\end{lemma}
\begin{proof}
(a) Using the definitions and \eqref{eqnBOvolhvol},
\begin{alignat*}{3}
\widehat{(Q\phi)}(k)
&=\vol_c\hskip-5pt\sum_{x\in\cX_\crs}\hskip-5pt 
                      (Q\phi)(x) e^{-i k\cdot x }
&=\vol_f\vol_c
    \hskip-4pt\sum_{\atop{x\in\cX_\crs}{ u\in\cX_\fin}}\hskip-4pt 
                       e^{-i k\cdot x }q(x-u)\phi(u)\displaybreak[0]\\
&\hskip-20pt=\sfrac{\vol_f}{|\cX_\crs|}
 \hskip-5pt\sum_{\atop{\atop{x\in\cX_\crs}{u\in\cX_\fin}}
                       {p\in\hat\cX_\fin}} \hskip-4pt
                       e^{-i k\cdot x }
                       e^{i u\cdot p }q(x-u)\hat\phi(p)\ 
&=\sfrac{\vol_f}{|\cX_\crs|}
  \hskip-5pt\sum_{\atop{\atop{x\in\cX_\crs}{ u\in\cX_\fin}}
                       {p\in\hat\cX_\fin}} \hskip-4pt
                       e^{-i k\cdot x }
                       e^{i (x-u)\cdot p }q(u)\hat\phi(p)\displaybreak[0]\\
&\hskip-20pt=\sfrac{1}{|\cX_\crs|}
             \sum_{\atop{x\in\cX_\crs}{p\in\hat\cX_\fin}}
                       e^{-i (k-p)\cdot x }
                       \hat q(p)\hat\phi(p)
&=\sfrac{1}{|\cX_\crs|}
             \sum_{\atop{x\in\cX_\crs}{p\in\hat\cX_\fin}}
                       e^{-i (k-\hat\pi(p))\cdot x }
                       \hat q(p)\hat\phi(p)\\
&\hskip-20pt= \sum_{ \atop{p\in\hat\cX_\fin}{\hat\pi(p)=k} }
                       \hat q(p)\hat\phi(p)
\end{alignat*}
The computation for $\widehat{(Q^\star\psi)}(p)$ is similar.
For the composite operators
\begin{align*}
\widehat{(QQ^\star\psi)}(k)
&=\sum_{ \atop{p\in\hat\cX_\fin}{ \hat\pi(p)=k} }
                       \hat q(p)\widehat{(Q^\star\psi)}(p)
=\sum_{ \atop{p\in\hat\cX_\fin}{ \hat\pi(p)=k} }
                       \hat q(p)\overline{\hat q(p)}\,\hat\psi\big(\hat\pi(p)\big)
=\sum_{ \atop{p\in\hat\cX_\fin}{ \hat\pi(p)=k} }
                       \big|\hat q(p)\big|^2\,\hat\psi(k)
\end{align*}
and similarly for $\widehat{(Q^\star Q\phi)}(p)$.

\Item (b) Since 
$$
a(\ru,\ru')=\vol_c\sum_{\rx\in\cZ_\crs}q(\rx-\ru)q(\rx-\ru')
$$
we have
\begin{align*}
\hat a_\rk(\ell,\ell')
&=\sfrac{\vol_f\vol_c}{|\cB|}
    \sum_{ \atop{\atop{[\ru]\in\cB}
                      {\ru'\in\cZ_\fin}}
                {\rx\in\cZ_\crs}}\hskip-5pt
      e^{-i(\rk+\ell)\cdot(\ru-\rx) } q(\rx-\ru)q(\rx-\ru') 
         e^{i(\rk+\ell')\cdot(\ru'-\rx) }
     \\
&=\vol_f^2
    \sum_{ \atop{\atop{[\ru]\in\cB}
                      {\ru'\in\cZ_\fin}}
                {\rx\in\cZ_\crs}}\hskip-5pt
      e^{-i(\rk+\ell)\cdot(\ru-\rx) } q(\rx-\ru)q(\ru') 
          e^{-i(\rk+\ell')\cdot\ru' }
     \\
&=\vol_f^2
    \sum_{\ru,\ru'\in\cZ_\fin}\hskip-5pt
      e^{i(\rk+\ell)\cdot\ru } q(\ru)q(\ru') 
          e^{-i(\rk+\ell')\cdot\ru' }
     \\
&=\overline{q(\rk+\ell)}\,q(\rk+\ell')
\end{align*}
\end{proof}

\begin{example}[Example \ref{exBOnaive}, continued]\label{exBOnaiveCont}
In the notation of Example \ref{exBOnaive}, the Fourier transform of $q$ is
\begin{align*}
\hat q(p)&=\sfrac{\vol_f}{\vol_c}
       \sum_{\atop{u\in\cX_\fin}{u\in\sq_0}}  e^{-i p\cdot u }
=u_{L_T}(\veps_T p_0)\prod_{\ell=1}^\sd u_{L_X}(\veps_X p_\ell)
\end{align*}
where
\begin{equation*}
u_L(\om)
=\sfrac{1}{L}\sum_{k=-\frac{L-1}{2}}^{\frac{L-1}{2}}e^{-i\om k}
=\begin{cases}\frac{1}{L}\frac{\sin\frac{L\om}{2}}{\sin\frac{\om}{2}}
                    & \text{if $\om\notin 2\pi\bbbz$}\\
                  \noalign{\vskip0.05in}
                1& \text{otherwise}
  \end{cases}
\end{equation*}

\centerline{\includegraphics{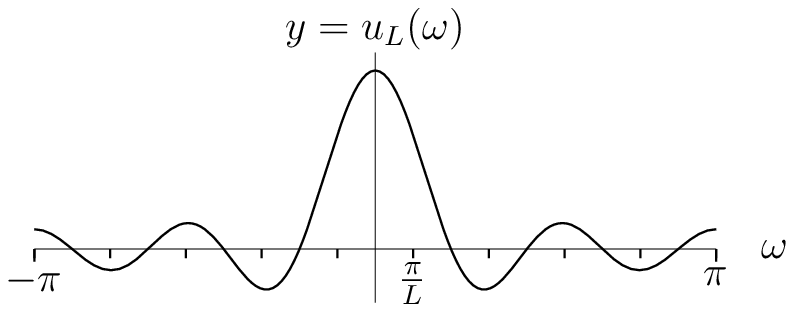}}

\end{example}

\begin{remark}\label{remBOlessnaive}
For the $q$ of Examples \ref{exBOnaive} and \ref{exBOnaiveCont}, 
which is, up to a multiplicative constant,
the characteristic function of a rectangle, the Fourier transform $\hat q(p)$
decays relatively slowly for large $p$. Choosing a smoother $q$ increases
the rate of decay of $\hat q(p)$. A convenient way to ``smooth off'' $Q$ is
to select an even $\fq\in\bbbn$ and choose $q$ to be the inverse Fourier
transform of
\begin{equation*}
\hat q(p)
=u_{L_T}(\veps_T p_0)^\fq\prod_{\ell=1}^\sd u_{L_X}(\veps_X p_\ell)^\fq
\end{equation*}
For example, when $\fq=2$, $q$ is the convolution of (a constant times)
the characteristic function of a rectangle with itself and so is a ``tent''
function. In \cite{PAR1,PAR2}, we use $\fq\ge 4$.
\end{remark}

\subsection{Analyticity of the Fourier Transform and $L^1$--$L^\infty$ Norms}
Define, for any $m\ge 0$ and $\ra:\bbbx\times\bbbx'\rightarrow\bbbc$,
with $\bbbx$ and $\bbbx'$ being any of our lattices,
\begin{equation*}
\|\ra\|_m
=\max\Big\{\sup_{y\in\bbbx}\,\vol_{X'}\!\!\!\sum_{y'\in\bbbx'} 
         e^{m|y-y'|}|\ra(y,y')|\ ,\ 
\sup_{y'\in\bbbx'}\vol_{X}\!\!\sum_{y\in\bbbx} e^{m|y-y'|}|\ra(y,y')|
\Big\}
\end{equation*}
Here $\vol_X$ and $\vol_{X'}$ is the volume of a single cell in $X$ and
$X'$, respectively.

\begin{lemma}\label{lemBOlonelinfty}
Let $a(\ru,\ru'):\cZ_\fin\times\cZ_\fin\rightarrow\bbbc$, 
$b(\ru,\rx):\cZ_\fin\times\cZ_\crs\rightarrow\bbbc$ and
$c(\rx,\ru):\cZ_\crs\times\cZ_\fin\rightarrow\bbbc$ be
translation invariant with respect to $\cZ_\crs$ and have finite
$L^1$--$L^\infty$ norms. Let $0<m''<m'<m$.
\begin{enumerate}[label=(\alph*), leftmargin=*]
\item 
If $\|a\|_m<\infty$, then, for each 
$\ell,\ell'\in\hat\cB$,  
$\hat a_\rk(\ell,\ell')$ is analytic in $|\Im \rk|<m$ and
\begin{equation*}
\sup_{|\Im \rk|<m} \big|\hat a_\rk(\ell,\ell')\big| \le \|a\|_m
\end{equation*}
\item
If, for each $\ell,\ell'\in\hat\cB$,  
$\hat a_\rk(\ell,\ell')$ is analytic in $|\Im \rk|<m$, then,
\begin{alignat*}{3}
\sup_{\ru,\ru'\in\cZ_\fin}\big|a(\ru,\ru')\big|e^{m'|\ru-\ru'|}
& \le \sfrac{1}{\vol_c} \sup_{|\Im\rk|=m'}\sum_{\ell,\ell'\in\hat\cB}
      \big| \hat a_{\rk}(\ell,\ell')\big|&
&\le \sfrac{|\cB|}{\vol_f} \sup_{\atop{|\Im\rk|=m'}{\ell,\ell'\in\hat\cB}}
      \big| \hat a_{\rk}(\ell,\ell')\big|\\
\|A\|_{m''}\le\|a\|_{m''}
&\le \sfrac{C_{m'-m''}}{\vol_c}  \hskip-3pt\sup_{|\Im\rk|=m'}
\hskip-2pt\sum_{\ell,\ell'\in\hat\cB}\hskip-7pt
     \big| \hat a_{\rk}(\ell,\ell')\big| &
&\le \sfrac{C_{m'-m''}|\cB|}{\vol_f}\hskip-5pt
    \sup_{\atop{|\Im\rk|=m'}{\ell,\ell'\in\hat\cB}}\hskip-9pt
      \big| \hat a_{\rk}(\ell,\ell')\big|
\end{alignat*}
where $A$ is the periodization of $a$ and
$
C_{m'-m''}=\vol_f\sum_{\ru\in\cZ_\fin}  e^{-(m'-m'')|\ru|}
$.

\item 
If, for each $\ell\in\hat\cB$,  
$\hat b_\rk(\ell)$ is analytic in $|\Im \rk|<m$, then,
\begin{align*}
\sup_{\atop{\ru\in\cZ_\fin}{\rx\in\cZ_\crs}}
\big|b(\ru,\rx)\big|e^{m'|\ru-\rx|}
& \le \sfrac{1}{\vol_c}  \sup_{|\Im\rk|=m'}  \sum_{\ell\in\hat\cB}
    \big| \hat b_{\rk}(\ell)\big|
\le \sfrac{1}{\vol_f} \sup_{\atop{|\Im\rk|=m'}{\ell\in\hat\cB}}
      \big| \hat b_{\rk}(\ell)\big|
\end{align*}
If, for each $\ell'\in\hat\cB$,  
$\hat c_\rk(\ell')$ is analytic in $|\Im \rk|<m$, then,
\begin{align*}
\sup_{\atop{\ru\in\cZ_\fin}{\rx\in\cZ_\crs}}
\big|c(\rx,\ru)\big|e^{m'|\rx-\ru|}
& \le \sfrac{1}{\vol_c} \sup_{|\Im\rk|=m'} \sum_{\ell'\in\hat\cB}
      \big| \hat c_{\rk}(\ell')\big|
\le \sfrac{1}{\vol_f} \sup_{\atop{|\Im\rk|=m'}{\ell'\in\hat\cB}}
      \big| \hat c_{\rk}(\ell')\big|
\end{align*}
\end{enumerate}
\end{lemma}
\begin{proof} (a)
If $|\Im\rk|<m$, then
\begin{align*}
\big|\hat a_\rk(\ell,\ell')\big| 
\le \sfrac{\vol_f}{|\cB|}
     \sum_{\atop{[\ru]\in\cB}{\ru'\in\cZ_\fin}}
        \hskip-5pt
      \big|a(\ru,\ru')\big|e^{m|\ru- \ru'|} 
\le \sfrac{1}{|\cB|}
    \sum_{[\ru]\in\cB} \|a\|_m
\le \|a\|_m
\end{align*}
Analyticity in $\rk$ follows from the uniform convergence of the series
on $|\Im \rk|<m$.

\Item (b) Fix any $\ru,\ru'\in\cZ_\fin$. Set
$\rq=m'\sfrac{\ru-\ru'}{|\ru-\ru'|}$. Then
\begin{align*}
a(\ru,\ru')e^{m'|\ru-\ru'|}
&= \sum_{\ell,\ell'\in\hat\cB}
\int_{\hat\cZ_\crs}\hskip-10pt
                \hat a_\rk(\ell,\ell') e^{i(\rk-i\rq)\cdot (\ru-\ru') }
                e^{i\ell\cdot \ru } e^{-i\ell'\cdot \ru' }
                \sfrac{d^{1+\sd}\rk}{(2\pi)^{1+\sd}}\cr
&= \sum_{\ell,\ell'\in\hat\cB}
\int_{\hat\cZ_\crs}\hskip-10pt
                \hat a_{\rk+i\rq}(\ell,\ell') e^{i\rk\cdot (\ru-\ru') }
                e^{i\ell\cdot \ru } e^{-i\ell'\cdot \ru' }
                \sfrac{d^{1+\sd}\rk}{(2\pi)^{1+\sd}}
\end{align*}
where we have applied Stokes' theorem, using analyticity in $\rk$ and the fact that
\begin{equation*}
e^{i\rk\cdot (\ru-\ru') }
   \sum_{\ell,\ell'\in\hat\cB}\hat a_\rk(\ell,\ell') 
                e^{i\ell\cdot \ru } e^{-i\ell'\cdot \ru' }
=e^{i\rk\cdot(\ru-\ru')}a_\rk(\ru,\ru')
\end{equation*}
is periodic in the real part
of $\rk$ with respect to $\sfrac{2\pi}{\veps_TL_T}\bbbz
   \times\sfrac{2\pi}{\veps_XL_X}\bbbz^\sd$.
Hence
\begin{align*}
\big|a(\ru,\ru')\big|e^{m'|\ru-\ru'|}
&\le \int_{\hat\cZ_\crs}\sum_{\ell,\ell'\in\hat\cB}
               \big| \hat a_{\rk+i\rq}(\ell,\ell')\big| 
                \sfrac{d^{1+\sd}\rk}{(2\pi)^{1+\sd}}\cr
&\le \sfrac{1}{\vol_c} \sup_{\rk\in\hat\cZ_\crs}
     \sum_{\ell,\ell'\in\hat\cB}
     \big| \hat a_{\rk+i\rq}(\ell,\ell')\big|\\
&\le \sfrac{|\cB|}{\vol_f}
     \sup_{\atop{\rk\in\hat\cZ_\crs}{
            \ell,\ell'\in\hat\cB}}
      \big| \hat a_{\rk+i\rq}(\ell,\ell')\big|
\end{align*}
The second bound is obvious from
\begin{align*}
\vol_f\sum_{y'\in\cX_\fin}\big|A\big([\ru],y'\big)\big|
                                      e^{m''|[\ru]-y'|}
&=\vol_f\sum_{y'\in\cX_\fin}\bigg|
      \smsum_{\atop{\ru'\in\cZ_\fin}{[\ru']=y'}}a(\ru,\ru')\bigg|
      e^{m''|[\ru]-y'|}\\
&\le\vol_f\sum_{\ru'\in\cZ_\fin}\big|a(\ru,\ru')\big|
      e^{m''|\ru-\ru'|}
\end{align*}
(with the distance $|[\ru]-y'|$ measured in $\cX_\fin$ and the
distance $|\ru-\ru'|$ measure in $\cZ_\fin$)
and the similar bound with the roles of $\ru$ and $\ru'$ interchanged.

\Item (c) The proof is much the same as that of part (b). 

\end{proof}

\begin{lemma}\label{lemBOuniqueness}
Let $m>0$. Let, for each $\ell,\ell'\in\hat\cB$,  
$\hat b_\rk(\ell,\ell')$ be analytic in $|\Im \rk|<m$. Assume that
\begin{equation*}
\hat b_{\rk+\rp}(\ell,\ell')=\hat b_{\rk}(\ell+\rp,\ell'+\rp)
\end{equation*}
for all $p\in \sfrac{2\pi}{\veps_TL_T}\bbbz
\times \sfrac{2\pi}{\veps_XL_X}\bbbz^\sd$ and 
$\ell,\ell'\in\hat\cB$. Set
\begin{align*}
a(\ru,\ru')&=\sum_{\ell,\ell'\in\hat\cB}
\int_{\hat\cZ_\crs}\hskip-10pt
                e^{i\ell\cdot \ru }\hat b_\rk(\ell,\ell')e^{-i\ell'\cdot \ru' }
                e^{i\rk\cdot (\ru-\ru') }  
                \sfrac{d^{1+\sd}\rk}{(2\pi)^{1+\sd}} 
\end{align*}
Then $a(\ru,\ru'):\cZ_\fin\times\cZ_\fin\rightarrow\bbbc$ 
obeys the conditions of Definition \ref{defBOperiodization} and
\begin{equation*}
\hat a_{\rk}(\ell,\ell')=\hat b_{\rk}(\ell,\ell')
\end{equation*}
for all $\rk\in\bbbr\times\bbbr^\sd$ and $\ell,\ell'\in\hat\cB$.
\end{lemma}
\begin{proof} The proof is straightforward.
\end{proof}  

\subsection{Functions of Periodic Operators}
Let $C$ be a simple, closed, positively oriented, piecewise smooth curve 
in the complex plane and denote by $\cO_C$ its interior. Denote by 
$\si(A)$ the spectrum of the bounded operator $A$ and assume 
$\si(A)\subset\cO_C$. Let $f(z)$ be analytic on the closure of $\cO_C$. 
Then, by the Cauchy integral formula,
\begin{equation}\label{eqnBOfofA}
f(A)=\sfrac{1}{2\pi i}\oint_{C}\sfrac{f(\ze)}{\ze\bbbone -A}\,d\ze
\end{equation}
and, for any $m\ge 0$,
\begin{equation}\label{eqnBOfofAbnd}
\|f(A)\|_m\le\sfrac{1}{2\pi}|C|
     \sup_{\ze\in C}|f(\ze)|
     \sup_{\ze\in C}\big\|(\ze\bbbone-A)^{-1}\big\|_m
\end{equation}

\begin{lemma}\label{lemBOfnbnd}
Let 
\begin{itemize}[leftmargin=*, topsep=2pt, itemsep=0pt, parsep=0pt]
\item
$a(\ru,\ru'):\cZ_\fin\times\cZ_\fin\rightarrow\bbbc$ 
obey the conditions of Definition \ref{defBOperiodization},
\item
$C$ be a simple, closed, positively oriented, piecewise
smooth curve in the complex plane with interior $\cO_C$,
\item
$\cO$ contain the closure of $\cO_C$ and $f:\cO\rightarrow\bbbc$ be analytic, and
\item
$0<m''<m'<m$.
\end{itemize}
Suppose that 
\begin{itemize}[leftmargin=*, topsep=2pt, itemsep=0pt, parsep=0pt]
\item
for each $\ell, \ell'\in\hat\cB$, $\hat a_\rk(\ell,\ell')$
is analytic in $|\Im \rk|<m$.
\item
for each 
$\ze\in\bbbc\setminus\cO_C$ and each $\rk$ with $|\Im \rk|<m$, the matrix 
$\big[\ze\de_{\ell,\ell'}-\hat a_\rk(\ell,\ell')\big]_{\ell, \ell'
                                                    \in\hat\cB}$ 
is invertible.
\end{itemize}
Denote by $A$ the periodization of $a$.  Then $f(A)$, defined
by \eqref{eqnBOfofA}, exists and
\begin{align*}
\|f(A)\|_{m''}
&\le \sfrac{C_{m'-m''}}{2\pi\ \vol_c} |C|\sup_{\ze\in C}|f(\ze)|
     \sup_{\atop{|\Im\rk|=m'}{\ze\in C}} 
     \sum_{\ell,\ell'\in\hat\cB}
     \big|(\ze\bbbone-\hat a_{\rk})^{-1}(\ell,\ell')\big|\\
&\le \sfrac{C_{m'-m''}|\cB|}{2\pi\ \vol_f}
     |C|\sup_{\ze\in C}|f(\ze)|
    \sup_{\atop{|\Im\rk|=m'}
               {\atop{\ell,\ell'\in\hat\cB}{\ze\in C }}
         }
      \big|(\ze\bbbone-\hat a_{\rk})^{-1}(\ell,\ell')\big|
\end{align*} 
Here $(\ze\bbbone-\hat a_{\rk})^{-1}$ refers to the inverse of the
$|\hat\cB|\times |\hat\cB|$ matrix
$\big[\ze\de_{\ell,\ell'}-\hat a_\rk(\ell,\ell')\big]_{\ell, \ell'
                                                    \in\hat\cB}$.
\end{lemma}
\begin{proof}
Each matrix element of $\ze\bbbone-\hat a_k$ is
continuous on 
\begin{equation*}
\cD=\set{(\ze,\rk)\in\bbbc^2}{\ze\in\bbbc\setminus\cO_C,\ |\Im\rk|<m}
\end{equation*}
Furthermore $\det(\ze\bbbone-\hat a_k)$ does not vanish on $\cD$.
Hence every matrix element of $\big(\ze\bbbone-\hat a_k\big)^{-1}$
is also continuous on $\cD$ and in particular is bounded on compact subsets
of $\cD$.
Set, for each for each $\ze\in\bbbc\setminus\cO_C$, and 
$\ru,\ru'\in\cZ_\fin$,
\begin{equation*}
r_\ze(\ru,\ru')=
\sum_{\ell,\ell'\in\hat\cB}
\int_{\hat\cZ_\crs}\hskip-10pt
                e^{i\ell\cdot \ru }
                (\ze\bbbone-\hat a_{\rk})^{-1}(\ell,\ell')
                e^{-i\ell'\cdot \ru' }
                e^{i\rk\cdot (\ru-\ru') }  
                \sfrac{d^{1+\sd}\rk}{(2\pi)^{1+\sd}} 
\end{equation*}
By Lemma \ref{lemBOuniqueness}, $r_\ze(\ru,\ru')$ obeys the conditions of 
Definition \ref{defBOperiodization} and
\begin{equation*}
\hat r_{\ze,\rk}(\ell,\ell') = (\ze\bbbone-\hat a_{\rk})^{-1}(\ell,\ell')
\end{equation*}
By Lemma \ref{lemBOperiodalg}, $r_\ze=(\ze\bbbone- a)^{-1}$, as operators on
$L^2(\cZ_\fin)$.
By  Remark \ref{remBOperiodization}.c, for each $\ze\in\bbbc\setminus\cO_C$, 
the periodization of $r_\ze(\ru,\ru')$ is $\big(\ze\bbbone-A\big)^{-1}$. 
In particular, $\si(A)\subset\cO_C$. By Lemma \ref{lemBOlonelinfty}.b,
\begin{align*}
\|(\ze\bbbone-A)^{-1}\|_{m''}
&\le \sfrac{C_{m'-m''}}{\vol_c} \sup_{|\Im\rk|=m'}
    \sum_{\ell,\ell'\in\hat\cB}
      \big|(\ze\bbbone-\hat a_{\rk})^{-1}(\ell,\ell')\big|\\
&\le \sfrac{C_{m'-m''}|\cB|}{\vol_f}
    \sup_{|\Im\rk|=m'\atop\ell,\ell'\in\hat\cB}
      \big|(\ze\bbbone-\hat a_{\rk})^{-1}(\ell,\ell')\big|
\end{align*}
Then, by \eqref{eqnBOfofAbnd},
\begin{align*}
\|f(A)\|_{m''}
&\le\sfrac{1}{2\pi}|C|\sup_{\ze\in C}|f(\ze)|
     \sup_{\ze\in C}\big\|(\ze\bbbone-A)^{-1}\big\|_{m''}\\
&\le \sfrac{C_{m'-m''}}{2\pi\ \vol_c} |C|\sup_{\ze\in C}|f(\ze)|
     \sup_{\atop{|\Im\rk|=m'}{\ze\in C}} 
     \sum_{\ell,\ell'\in\hat\cB}
     \big|(\ze\bbbone-\hat a_{\rk})^{-1}(\ell,\ell')\big|\\
&\le \sfrac{C_{m'-m''}|\cB|}{2\pi\ \vol_f}
     |C|\sup_{\ze\in C}|f(\ze)|
    \sup_{\atop{|\Im\rk|=m'}
               {\atop{\ell,\ell'\in\hat\cB}{\ze\in C }}
          }
      \big|(\ze\bbbone-\hat a_{\rk})^{-1}(\ell,\ell')\big|
\end{align*}
\end{proof}

\subsection{Scaling of Periodized Operators}
Scaling plays an important role in the construction of \cite{PAR1,PAR2}.
See, for example,  \cite[Definition \defHTscaling\ and \S\chapSCscaling]{PAR1}
and \cite[(\eqnPINTlft)]{POA}. For the current abstract setting,
select scaling factors $\si_T$ and $\si_X$ and define the scaled lattices
\begin{align*}
\cZ_\fin^{(s)}
   &= \sfrac{\veps_T}{\si_T}\bbbz\times\sfrac{\veps_X}{\si_X}\bbbz^\sd &
\vol_f^{(s)}&=\sfrac{\veps_T\veps_X^\sd}{\si_T\si_X^\sd}
\\
\cZ_\crs^{(s)}&=L_T\sfrac{\veps_T}{\si_T}\bbbz
                \times L_X\sfrac{\veps_X}{\si_X}\bbbz^\sd &
\vol_c^{(s)}&=\sfrac{(L_T\veps_T)(L_X\veps_X)^\sd}{\si_T\si_X^\sd}
\\
\hat\cZ_\crs^{(s)}&=\big(\bbbr/\sfrac{2\pi\si_T}{\veps_TL_T}\bbbz\big)
\times \big(\bbbr^\sd/\sfrac{2\pi\si_X}{\veps_XL_X}\bbbz^\sd\big)
\\
\cB^{(s)}
  &=\big(\sfrac{\veps_T}{\si_T}\bbbz/L_T\sfrac{\veps_T}{\si_T}\bbbz\big)\times
    \big(\sfrac{\veps_X}{\si_X}\bbbz^\sd/L_X\sfrac{\veps_X}{\si_X}\bbbz^\sd)
\cong \cX_\fin^{(s)}/\cX_\crs^{(s)}
\\
\hat\cB^{(s)}&= 
\big(\sfrac{2\pi\si_T}{L_T\veps_T}\bbbz/\sfrac{2\pi\si_T}{\veps_T}\bbbz\big)
   \times
\big(\sfrac{2\pi\si_X}{L_X\veps_X}\bbbz^\sd/
                \sfrac{2\pi\si_X}{\veps_X}\bbbz^\sd\big)
\end{align*}
The map  $\bbbl(\tau,\bx)=(\si_T\tau,\si_X\bx)$ gives bijections
\begin{equation*}
\bbbl:\cZ_\fin^{(s)}\rightarrow \cZ_\fin\qquad
\bbbl:\cZ_\crs^{(s)}\rightarrow \cZ_\crs\qquad
\bbbl:\cB^{(s)}\rightarrow \cB
\end{equation*}
$\bbbl$ induces linear bijections 
$\bbbl_*:L^2\big(\cZ_\fin^{(s)}\big)\rightarrow L^2\big(\cZ_\fin\big)$
and
$\bbbl_*:L^2\big(\cZ_\crs^{(s)}\big)\rightarrow L^2\big(\cZ_\crs\big)$
by $\bbbl_*(\al)(\bbbl \ru) = \al(\ru)$.
Observe that
\begin{align*}
\<\bbbl_*\al,\bbbl_*\be\>_f
= \si_T\si_X^\sd\<\al,\be\>_f^{(s)} \qquad
\<\bbbl_*\al,\bbbl_*\be\>_c
= \si_T\si_X^\sd\<\al,\be\>_c^{(s)} \qquad
\end{align*}

\begin{lemma}\label{lemPoPscaling}
Let $a:L^2\big(\cZ_\fin\big)\rightarrow L^2\big(\cZ_\fin\big)$ have 
kernel $a(\ru,\ru')$.
\begin{enumerate}[label=(\alph*), leftmargin=*]
\item 
The kernel of $\bbbl_*^{-1} a\,\bbbl_*$ is
\begin{equation*}
a^{(s)}(\rv,\rv')
=\si_T\si_X^\sd\,a\big(\bbbl\rv,\bbbl\rv'\big)
\end{equation*}
\item 
The Fourier transform of the kernel of $\bbbl_*^{-1} a\bbbl_*$ is
\begin{align*}
\hat a_\rk^{(s)}(\ell,\ell')
&=\hat a_{\bbbl^{-1}\rk}(\bbbl^{-1}\ell,\bbbl^{-1}\ell')
\qquad \text{for }\rk\in\bbbr\times\bbbr^\sd,
\quad \ell,\ell'\in\hat\cB^{(s)}
\end{align*}
\item  
If $m\ge\max\big\{\sfrac{1}{\si_T},\sfrac{1}{\si_X}\big\}m_s$,
then
$
\|a^{(s)}\|_{m_s} \le \|a\|_m
$.
\end{enumerate}
\end{lemma}
\begin{proof} 
(a) For $\al\in L^2\big(\cZ_\fin^{(s)}\big)$ and $\rv\in\cZ_\fin^{(s)}$,
\begin{align*}
(\bbbl_*^{-1} a\bbbl_*\,\al)(\rv)
&=\vol_f \sum_{\ru'\in\cZ_\fin}
           \,a(\bbbl\rv,\ru')\,\big(\bbbl_*\,\al\big)(\ru')\\
&=\vol_f \sum_{\ru'\in\cZ_\fin}
           \,a(\bbbl\rv,\ru')\,\al(\bbbl_*^{-1}\ru')\cr
&=\vol_f^{(s)} \sum_{\rv'\in\cZ_\fin^{(s)}}
           \si_T\si_X^\sd\,a(\bbbl\rv,\bbbl\rv')\,\al(\rv')
\end{align*}

\Item (b) 
By \eqref{eqnBOifkerkell} and part (a),
\begin{align*}
\hat a_\rk^{(s)}(\ell,\ell')
&=\sfrac{\vol_f}{|\hat\cB|}
    \sum_{\atop{[\rv]\in\cB^{(s)}}{\rv'\in\cZ_\fin^{(s)}}}  \hskip-5pt
      e^{-i\ell\cdot\rv } a\big(\bbbl\rv,\bbbl\rv'\big)e^{i\ell'\cdot\rv' }
      e^{-i\rk\cdot(\rv- \rv')} 
     \\
&=\sfrac{\vol_f}{|\hat\cB|}
    \sum_{\atop{[\ru]\in\cB}{\ru'\in\cZ_\fin}}  \hskip-5pt
     e^{-i(\bbbl^{-1}\ell)\cdot\ru } a(\ru,\ru')e^{i(\bbbl^{-1}\ell')\cdot\ru' }
      e^{-i(\bbbl^{-1}\rk)\cdot(\ru- \ru')} 
     \\
&=\hat a_{\bbbl^{-1}\rk}(\bbbl^{-1}\ell,\bbbl^{-1}\ell')
\end{align*}

\Item (c)
This part follows from the inequality
\begin{align*}
\sup_{\rv\in\cZ_\fin^{(s)}}\vol_f^{(s)}\!\!\sum_{\rv'\in\cZ_\fin^{(s)}} 
                             e^{m_s|\rv-\rv'|}|a^{(s)}(\rv,\rv')|
&=\sup_{\rv\in\cZ_\fin^{(s)}}\vol_f^{(s)}\!\!\sum_{\rv'\in\cZ_\fin^{(s)}} 
                     e^{m_s|\rv-\rv'|}\si_T\si_X^\sd|a(\bbbl\rv,\bbbl\rv')|\\
&=\sup_{\ru\in\cZ_\fin}\vol_f\!\!\sum_{\ru'\in\cZ_\fin} 
                  e^{m_s|\bbbl^{-1}\ru-\bbbl^{-1}\ru'|}|a(\ru,\ru')|\\
&\le\sup_{\ru\in\cZ_\fin}\vol_f\!\!\sum_{\ru'\in\cZ_\fin} 
                  e^{m|\ru-\ru'|}|a(\ru,\ru')|
\end{align*}
and the corresponding inequality with $\rv$ summed over and $\rv'$
supped over.
\end{proof}

More generally,

\begin{lemma}\label{lemPoPscalingCrs}
Let $b:L^2\big(\cZ_\crs\big)\rightarrow L^2\big(\cZ_\fin\big)$ and
$c:L^2\big(\cZ_\fin\big)\rightarrow L^2\big(\cZ_\crs\big)$ have 
kernels $b(\ru,\rx)$ and $c(\rx,\ru)$ respectively.
\begin{enumerate}[label=(\alph*), leftmargin=*]
\item 
The kernels of $\bbbl_*^{-1} b\,\bbbl_*$ and 
$\bbbl_*^{-1} c\,\bbbl_*$  are
\begin{align*}
b^{(s)}(\rv,\rx)
=\si_T\si_X^\sd\,b\big(\bbbl\rv,\bbbl\rx\big)\qquad
c^{(s)}(\rx,\rv)
=\si_T\si_X^\sd\,c\big(\bbbl\rx,\bbbl\rv\big)
\end{align*}
\item 
The Fourier transform of the kernels of 
$\bbbl_*^{-1} b\,\bbbl_*$ and $\bbbl_*^{-1} c\,\bbbl_*$ are
\begin{align*}
\hat b_\rk^{(s)}(\ell)
=\hat b_{\bbbl^{-1}\rk}(\bbbl^{-1}\ell) \qquad
\hat c_\rk^{(s)}(\ell')
=\hat c_{\bbbl^{-1}\rk}(\bbbl^{-1}\ell')
\qquad \text{for }\rk\in\bbbr\times\bbbr^\sd,
\quad \ell,\ell'\in\hat\cB^{(s)}
\end{align*}
\item 
If $m\ge\max\big\{\sfrac{1}{\si_T},\sfrac{1}{\si_X}\big\}m_s$,
then
\begin{align*}
\|b^{(s)}\|_{m_s} \le \|b\|_m\qquad
\|c^{(s)}\|_{m_s} \le \|c\|_m
\end{align*}
\end{enumerate}
\end{lemma}

\newpage
\bibliographystyle{plain}
\bibliography{refs}

\end{document}